\documentclass[12pt,a4paper]{article}
  \usepackage{times}
  \usepackage{fullpage}
  \usepackage{graphicx}
  \usepackage{amsthm}
  \usepackage{amsmath}
  \usepackage{amsfonts}
  \usepackage{amssymb}
  \usepackage{dsfont}
  \usepackage{setspace}
  \usepackage{algorithm}
  \usepackage{algorithmic}
  \usepackage{url}
  \usepackage{subfigure}
  \usepackage{comment}
  \usepackage{paralist}
  \usepackage{multirow}

  \IfFileExists{srcltx.sty} {\usepackage[active]{srcltx}}

  \DeclareMathOperator*{\argmin}{arg\,min}
  \DeclareMathOperator*{\argmax}{arg\,max}

  \newcommand{\norm}[1]{\left\Vert#1\right\Vert_1}
  \newcommand{\abs}[1]{\left\vert#1\right\vert}

  \newcommand{\ceil}[
  1]{
  \left\lceil #1 \right\rceil}
  \newcommand{\floor}[1]{\left\lfloor #1 \right\rfloor }
  \newcommand{\zbox}{\mathrm{ZBox}}
  \newcommand{\rbox}{
  \mathrm{RBox}}
  \newcommand{\girth}{\mathrm{girth}}%
  \newcommand{
  \avg}{\mathrm{avg}}
  \newcommand{\eqdf }{\triangleq}%
  \newcommand{\R }{\mathds{R}}
  \newcommand{\N }{\mathds {N}}
  
  \newcommand{\calN }{\mathcal{N}} 
  \newcommand{\calC }{\mathcal{C}} 
  \newcommand{\calT }{\mathcal{T}} 
  \newcommand{\calF }{\mathcal{F}}
  \newcommand{\calV }{\mathcal{V}} 

  \newcommand{\calP }{\mathcal{P}}

  \newcommand{\calX }{\mathbf{X}}
  \newcommand{\calO }{\mathcal{O}}
  \newcommand{\calE }{\mathcal{E}}

  \newcommand{\msp}{\textsc{min-sum-packing}}
  \newcommand{\msc}{\textsc{min-sum-covering}}
  \newcommand{\cip}{\textsc{CIP}}
  \newcommand{\clp}{\textsc{CLP}}
  \newcommand{\pip}{\textsc{PIP}}
  \newcommand{\plp}{\textsc{PLP}}

  \newcommand{\OPTLP}{\text{OPT}_\text{LP}}
  \newcommand{\OPTDP}{\text{OPT}_\text{DP}}

  \newcommand{\ignore}[1]{}

\newcommand{\TScapE}{\calT_S \cap \calE}
\newcommand{\TScapO}{\calT_S \cap \calO}
\newcommand{\vecc}[1]{\mathbf{#1}}
\newcommand{\vx}{\vecc{x}}
\newcommand{\vw}{\vecc{w}}
\newcommand{\vb}{\vecc{b}}
\newcommand{\vd}{\vecc{d}}
\newcommand{\vz}{\vecc{z}}
\newcommand{\vtx}{\vecc{\tilde{x}}}
\newcommand{\vy}{\vecc{\tilde{y}}}
\newcommand{\vA}{\vecc{A}}

  \newenvironment{proofof}[1]{\noindent \emph{ Proof of {#1}.}~}{\qed\endproof}

  \graphicspath{{figs/}}

  \newtheorem{definition}{Definition}
  \newtheorem{lemma}[definition]{Lemma}
  \newtheorem{claim}[definition]{Claim}
  \newtheorem{proposition}[definition]{Proposition}
  \newtheorem{theorem}[definition]{Theorem}
  \newtheorem{corollary}[definition]{Corollary}
  \newtheorem{observation}[definition]{Observation}

\begin{document}

  \title{Analysis of the Min-Sum Algorithm
    for Packing and Covering Problems
via Linear Programming}

  \author{
        Guy Even \thanks{School of Electrical Engineering, Tel-Aviv University, Tel-Aviv 69978, Israel.  \mbox{{E-mail}:\ {\tt guy@eng.tau.ac.il}.}}
        \and
        Nissim Halabi \thanks{School of Electrical Engineering, Tel-Aviv University, Tel-Aviv 69978, Israel. \mbox{{E-mail}:\ {\tt nissimh@eng.tau.ac.il}.}}}

  \date{}

   \maketitle

  \begin{abstract}
 Message-passing algorithms based on belief-propagation (BP) are successfully used in many applications including decoding error
 correcting codes and solving constraint satisfaction and inference problems.
 BP-based algorithms operate over graph representations, called factor graphs, that are used to model the
 input. Although in many cases BP-based algorithms exhibit impressive empirical results, not much has been proved when the factor graphs have cycles.

This work deals with packing and covering integer
programs in which the constraint matrix is zero-one, the
constraint vector is integral, and the variables are subject to box
constraints. We study the performance of the min-sum algorithm when
applied to the corresponding factor graph models of packing and covering LPs.

We compare the solutions computed by the min-sum algorithm for packing and covering problems to the optimal solutions of the corresponding linear programming (LP) relaxations.
In particular, we prove that if the LP has an optimal fractional solution, then for
each fractional component, the min-sum algorithm either computes multiple solutions or the solution oscillates below and above the fraction.
This implies that the
min-sum algorithm computes the optimal integral solution only if the
LP has a unique optimal solution that is integral.

The converse is not true in general. For a special case of packing and covering problems, we prove that if the LP has a unique optimal solution that is integral and on the boundary of the box constraints, then the min-sum algorithm computes the optimal solution in pseudo-polynomial time.

Our results unify and extend recent results for the maximum weight
matching problem by [Sanghavi \emph{et al.},'2011] and [Bayati \emph{et al.}, 2011] and for the maximum weight independent set problem [Sanghavi \emph{et al.}'2009].
\end{abstract}

  \section{Introduction} \label{sec:intro}

  We consider optimization problems over the integers called
  \emph{packing} and \emph{covering} problems.  Many optimization
  problems can be formulated as packing problems including maximum
  weight matchings and maximum weight independent sets.  Optimization
  problems such as minimum weight set-cover and minimum weight dominating
  set are special cases of covering problems. The input for both types
  of problems consists of an $m\times n$ zero-one constraint matrix $\vA$, an
  integral constraint vector $\vb$, an upper bound vector $\calX\in
  \N^n$, and a weight vector $\vw\in \R^n$. An integral vector $\vx\in
  \N^n$ is an \emph{integral packing} if $\mathbf{0} \leq \vx \leq
  \calX$ and $\vA\cdot \vx \leq \vb$.  In a packing problem the goal is
  to find an integral packing that minimizes $\vw^T \cdot \vx$.  An
  integral vector $\vx\in \N^n$ is an \emph{integral covering} if
  $\mathbf{0} \leq \vx \leq \calX$ and $\vA\cdot \vx \geq \vb$.  In a
  covering problem the goal is to find an integral packing that
  maximizes $\vw^T \cdot \vx$.

  Packing and covering problems generalize problems that are solvable
  in polynomial time (e.g., maximum matching) and problems that are
  NP-hard and even NP-hard to approximate (e.g., maximum independent
  set). The hardness of special cases of these problems imply that
  general algorithms for packing/covering problems are heuristic in
  nature. Two heuristics that are used in practice to solve such problems are
  linear programming (LP) and belief-propagation (BP).

  Linear programming deals with optimizing a linear function over
  polyhedrons (subsets of the Euclidean space
  $\R^n$)~\cite{bertsimas1997introduction}. Perhaps the most naive way to
  utilize linear programming in this setting is to solve the LP
  relaxation of the integer problem (i.e., relax the restriction that
  $\vx\in \N^n$ to the restriction $\vx\in\R^n$). If the result
  happens to be integral, then we are lucky and we have found an
  optimal integral packing or covering. A great deal of literature
  deals with characterizing problems for which this method works well
  (e.g., works on total unimodularity~\cite{schrijver1998theory}). In fact, LP
  decoding of error correcting codes works in the same fashion and has
  been proven to work well in average~\cite{FWK05,ADS09}.

  Belief-propagation is an algorithmic paradigm that deals with
  inference over graphical models~\cite{Pea88}. The graphical model
  that corresponds to packing/covering problems is a bipartite graph
  that represents the zero-one matrix $\mathbf{A}$. We focus on a
  common variant of belief-propagation that is called the min-sum
  algorithm (or the max-product algorithm).  In the variant we
  consider, the initial messages are all zeros and messages are not
  attenuated.

  Our main result is a proof that the min-sum algorithm is not better
  than the heuristic based on linear programming. This proof holds for
  every instance of packing and covering problems described above with
  respect to the min-sum algorithm with zero initialization and no
  attenuation.

\paragraph{Previous Work.}
Message-passing algorithms based on \emph{belief-propagation} (BP)
have been invented multiple times
(see~\cite{Gal63,Vit67,Pea88}). Numerous papers report empirical
results that demonstrate the usefulness of these algorithms for
decoding error correcting codes, inference with noise, constraint
satisfaction problems, and many other applications~\cite{Y11}. It took
a while until it was noticed that algorithms for decoding of Turbo
codes~\cite{BGT93} and LDPC codes~\cite{Gal63} are special variants of
BP~\cite{MMC98,Wib96}.

In this paper we focus on a common variant of BP called the
\emph{min-sum algorithm}, and consider the case where messages are
initialized to zero.  The BP algorithm is a message-passing algorithm
in which messages are sent along edges of a graph called the
\emph{factor graph}.  The factor graph of packing/covering problems is
a bipartite graph that represents that constraint matrix $\mathbf{A}$.
In essence, value computed by the min-sum algorithm for $x_i$ equals
the outcome of a dynamic programming algorithm over a path-prefix tree
rooted at the vertex corresponding to the $i$th column of
$\mathbf{A}$. Since dynamic programming computes an optimal solution
over trees, the min-sum algorithm is optimal when the factor graph is
a tree~\cite{Pea88,Wib96}. A major open problem is to analyze the
performance of BP (or even the min-sum algorithm) when the factor
graph is not a tree. Execution of algorithms based on BP over graphs
that contain cycles is often referred to as \emph{loopy BP}.

Recently, a few papers have studied the usefulness of the min-sum
algorithm for solving optimization problems compared to linear
programming.  Such a comparison for the maximum weight matching
problem appears in~\cite{BSS08,BBCZ11,SMW11} with respect to
constraints of the form $\sum_{u \text{ neighbor of } v}  x_{(u,v)} \leq 1$ for every vertex
$v$.  Loosely speaking, the main result that they show for maximum
weighted matching is that the min-sum algorithm is successful if and only if
the LP heuristic is successful.  The sufficient
condition states that if the LP relaxation has a
unique optimal solution and that solution is integral, then the
min-sum algorithm computes this solution in pseudo-polynomial time.
The necessary condition states that if the LP
relaxation has a fractional optimal solution, then the min-sum
algorithm fails.  In~\cite{SSW09}, the min-sum algorithm for the
maximum weighted independent set problem was studied with
respect to constraints $x_u+x_v\leq 1$ for every edge $(u,v)$. They
prove an analogous necessary condition and present a counter-example
that proves that the sufficient condition (i.e., unique optimal
solution for the LP that is integral) does not imply the success of
the min-sum algorithm.  The performance of the min-sum algorithm has
been also studied for computing shortest $s$-$t$ paths~\cite{RT08} and
min-cost flows~\cite{GSW12,BCMR13}.

\paragraph{Our Results.}
The results in this paper extend and generalize previous necessary
conditions for the success of the min-sum algorithm.  This necessary
condition implies that, compared to the LP heuristic,
the min-sum algorithm is not a better heuristic for solving packing
and covering problems. Our contributions can be summarized as follows:
(1)~We consider a unified framework of packing and covering problems.
Previous works deal with a zero-one constraint matrix $\mathbf{A}$
that has two nonzero entries in each column (for maximum weight
matching~\cite{BSS08,BBCZ11,SMW11}) or two nonzero entries in each row
(for maximum weight independent set~\cite{SSW09}). Our results hold with
respect to any zero-one constraint matrix $\mathbf{A}$.
(2)~We allow box constraints, namely $x_i \in\{0,1,\ldots,
  X_i\}$.  Previous results consider only zero-one variables.
(3)~Our oscillation results hold also when the LP
  relaxation has multiple solutions. To obtain such a result, we
  consider the set of optimal values computed by the min-sum algorithm
  at each variable (rather than declare failure if there are multiple
  optimal values). We compare these sets with the optimal solutions of
  the LP relaxation and show a weak oscillation
  between even and odd iterations.
  (4)~The analogous result for covering LPs is obtained by a simple
  reduction (see Claim~\ref{claim:CLPtoPLP}) that applies
  complementation; this reduction generalizes reductions from maximum
  matchings to minimum edge covers~\cite{SMW11}.
  (5)~We present a unified proof method based on graph covers. This
  method also enables us to prove convergence of the min-sum algorithm
  under certain restrictions (see Theorem~\ref{thm:msp-convergence} in Appendix~\ref{app:convergence_proof}).

\paragraph{Techniques.}
The main challenge in comparing between the min-sum algorithm and
linear programming is in finding a common structure that captures both
algorithms. It turns out that graph covers are a common structure.
Graph covers have been used previously to analyze iterative message
passing algorithms~\cite{VK05,EH11}. In the context of optimization
problems, $2$-covers have been used in~\cite{BBCZ11} to reduce
matchings in general graphs to matchings in bipartite
graphs~\cite{BSS08}.  Bayati \emph{et al.}~\cite{BBCZ11} write that
their ``proof gives a better understanding of the often-noted but
poorly understood connection between BP and LP through graph covers.''
We further clarify this connection by using higher order covers that
capture fractional optimal LP solutions, as suggested by Ruozzi and
Tatikonda~\cite{RT12}.  Graph covers not only capture LP solutions but
also solutions computed by the min-sum algorithm. In fact, the min-sum
algorithm performs the same computation over any graph cover because
it operates over a path-prefix tree of the factor graph. Hence we make
the mental experiment in which the min-sum algorithm is executed over
a graph cover in which all the basic feasible solutions are integral.
We avoid the problems associated with loopy-BP by considering a graph
cover, the girth of which is much larger than the number of iterations
of the min-sum algorithm.  Thus the execution of the min-sum algorithm
is equivalent to a dynamic programming over subtrees induced by balls
in the graph cover. This mental game justifies a dynamic programming
interpretation of the outcome of the min-sum algorithm. The dynamic
programming algorithm makes a ``local'' decision based on balls, the
radius of which is twice the number of iterations. The LP solution, on
the other hand, is a global solution.

The proof  proceeds by creating ``hybrid'' solutions that
either refute the (global) optimality of the LP solution or the
(local) optimality of the dynamic programming solution over the ball
in the graph cover.

  \section{Preliminaries} \label{sec:prelim}

  \subsection{Graph Terminology and Algebraic Notation}
  \paragraph{Algebraic Notation.}
  We denote vectors in bold, e.g., $\vx,\vz$.  We denote the $i$th
  coordinate of $\vx$ by $x_i$, e.g., $\vx = (x_1,\ldots,x_n)$.  For a vector
  $\vx\in\R^n$, let $\norm{\vx}\triangleq\sum_i\abs{x_i}$ denote
  the $\ell_1$ norm of $\vx$. The cardinality of a set $S$ is denoted
  by $\abs{S}$.  We denote by $[n]$ the set $\{0,1,2,...,n\}$ for $n
  \in \N$.  For a set $S \subseteq [n]$, we denote the projection of the vector $\vx$ onto indices in
  $S$ by $\vx_S \in
  \R^{\abs{S}}$ .

  A vector is \emph{rational} if all its components are rational.
  Similarly, a vector is \emph{integral} if all its components are
  integers. A vector is \emph{fractional} if it is not integral, i.e.,
  at least one of its components is not an integer.

  Let ${\calX}\in\N^n$ denote a non-negative integral vector. Denote the
  Cartesian product $[{X}_1]\times \cdots \times [{X}_n]$ by $\zbox({\calX})$.
  Similarly, denote the Cartesian product $[0,{X}_1]\times \cdots
  [0,{X}_n]$ by $\rbox({\calX})$. Note that vectors in $\zbox(\calX)$ are integral.

  \paragraph{Graph Terminology.}
  Let $G=(V,E)$ denote an undirected simple graph.
  Let $\calN_G(v)$ denote the set of neighbors of vertex $v \in V$ (not including $v$ itself).
  Let $\deg_G(v)$ denote the edge degree of vertex $v$
  in a graph $G$, i.e., $\deg_G(v)\triangleq\lvert\calN_G(v)\rvert$.
  For a set $S \subseteq V$ let
  $\mathcal{N}_G(S)\triangleq\bigcup_{v \in S}\mathcal{N}_G(v)$.
  A \emph{path} in $G$ is a sequence of vertices such that
  there exists an edge between every two consecutive vertices in the
  sequence.
A \emph{backtrack} in a path is a subpath that is a
    loop consisting of two edges traversed in opposite directions, i.e., a subsequence $(u,v,u)$. 
All  paths considered in this paper do not include backtracks.
The \emph{length} of a path is the number of edges in the path.
We denote the length of a path $p$ by $\lvert p\rvert$.
  Let $d_G(r,v)$ denote the distance
  (i.e., length of a shortest path) between vertex $r$ and $v$ in $G$,
  and let $\girth(G)$ denote the length of the shortest cycle in $G$.
  Let $B_G(v,t)$ denote the set of vertices in $G$ with distance at most $t$ from $v$, i.e., $B_G(v,t)\triangleq\{u\in V\mid d_G(v,u)\leq t\}$.

  The \emph{subgraph of $G$ induced by $S \subseteq V$} consists of
  $S$ and all edges in $E$, both endpoints of which are contained in
  $S$. Let $G_S$ denote the subgraph of $G$ induced by $S$.  A subset
  $S\subseteq V$ is an \emph{independent set} if there are no edges in
  the induced subgraph $G_S$.
  A graph $G=(V,E)$ is \emph{bipartite} if $V$ is the union of two
  disjoint nonempty independent sets.

  \subsection{Covering and Packing Linear Programs}
  We consider two types of linear programs called covering and packing
  problems. In both cases the matrices are zero-one matrices and the
  constraint vectors are positive.

  \noindent \medskip In the sequel we refer to the constraints
  $\vx\in\zbox(\calX)$ and $\vx\in\rbox(\calX)$ as \emph{box constraints}.

  \begin{definition}\label{def:C-P_IP_LP}
  Let $\mathbf{A}\in \{0,1\}^{m\times n}$ denote a zero-one matrix with $m$ rows and $n$ columns.
  Let $\vb\in \R_+^m$ denote a constraint vector, let $\vw\in \R^n$ denote a weight vector, and let $\calX\in\N^n$ denote a domain boundary vector.
      \begin{enumerate}
      \item[{$[\pip]$}] The integer program $\argmax\big\{\vw^T\cdot \vx~\big\vert~\mathbf{A}\cdot \vx \leq \vb,~\vx\in\zbox(\calX)\big\}$ is called a \emph{packing IP}, and denoted by \pip.
      \item[{$[\cip]$}] The integer program $\argmin\big\{\vw^T\cdot \vx~\big\vert~\mathbf{A}\cdot \vx \geq \vb,~\vx\in\zbox(\calX)\big\}$ is called a \emph{covering IP}, and denoted by \cip.
      \item[{$[\plp]$}] The linear program $\argmax\big\{\vw^T\cdot \vx~\big\vert~\mathbf{A}\cdot \vx \leq \vb,~\vx\in\rbox(\calX)\big\}$ is called a \emph{packing LP}, and denoted by \plp.
      \item[{$[\clp]$}] The linear program $\argmin\big\{\vw^T\cdot \vx~\big\vert~\mathbf{A}\cdot \vx \geq \vb,~\vx\in\rbox(\calX)\big\}$ is called a \emph{covering LP}, and denoted by \clp.
      \end{enumerate}
    \end{definition}

    \subsection{Factor Graph Representation of Packing and Covering LPs}
    The belief-propagation algorithm and its variant called the min-sum
    algorithm deal with graphical models known as \emph{factor graphs}
    (see, e.g., \cite{KFL01}).  In this section we review the definition of factor
    graphs that are used to model covering and packing problems.

  \begin{definition}[factor graph model of packing problems]\label{def:FG_model}
    A quadruple $\langle G,\Psi,\Phi, \calX\rangle$ is the \emph{factor
      graph model} of \pip\ if:
  \begin{itemize}

  \item $G=(\calV\cup \calC,E)$ is a bipartite graph that represents the
    zero-one matrix $\mathbf{A}$. The set of \emph{variable vertices}
    $\calV=\{v_1,\ldots,v_n\}$ corresponds to the columns of $\mathbf{A}$, and
    the set of \emph{constraint vertices} $\calC=\{C_1,\ldots,C_m\}$
    corresponds to the rows of $\mathbf{A}$. The edge set is defined by
    $E\triangleq\{(v_i,C_j)\mid \mathbf{A}_{ji}=1\}$.

  \item The vector $\calX\in \N^n$ defines the alphabets that are
    associated with the variable vertices.  The alphabet associated
    with $v_i$ equals $\{0,\ldots,X_i\}$.

  \item For each constraint vertex $C_j$, we define a packing factor function
$\psi_{C_j}:\zbox(\calX)\rightarrow\{0,-\infty\}$,
defined by
      \begin{align}\label{eqn:packing_factor}
      \psi_{C_j}(\mathbf{y}) &\eqdf
          \begin{cases}
            0 & \text{if } ~\sum_{v_i\in\calN_G(C_j)} y_i \leq  b_j\\
            -\infty & \text{otherwise}
          \end{cases}
      \end{align}
We denote the set of factor functions $\{\psi_{C_j}\}_j$ by $\Psi$.
\item For each variable vertex $v_i$, we define a variable function
  $\phi_{v_i}:[0,X_i]\rightarrow \R$ defined by $\phi_{v_i}(\beta)\triangleq w_i\cdot
  \beta$.
We denote the set of variable functions $\{\phi_{v_i}\}_i$ by $\Phi$.
  \end{itemize}
  \end{definition}
  We note that (1)~One could define a factor graph for \plp; the only
  difference is that the alphabet associated with variable vertex
  $v_i$ is the real interval $[0,X_i]$, and the range of each factor
  function is $\rbox(\calX)$. (2)~The factor functions are local in the
  sense that each constraint vertex $C_j$ can evaluate the value of
  $\psi_{C_j}$ based on the values of its neighbors.

  A vector $\vx\in\R^n$ is viewed as an \emph{assignment} to variable vertices in $\calV$
  where $x_i$ is assigned to vertex $v_i$.
To avoid composite indices, we use $x_v$ to denote the value assigned to $v$ by the assignment $\vx$.
An integral assignment $\vx$ is
  \emph{valid} if it satisfies all the constraints, namely, $x_i\in [X_i]$ for every $i$ and $\mathbf{A}\cdot \vx \leq \vb$.

  The factor graph model allows for the following equivalent formulation of the packing integer program:
  \begin{equation}
  \argmax\big\{\sum_{v\in\calV}\phi_v(x_v)+\sum_{C\in\calC}\psi_C(\vx)~\big\vert~
  \vx \in \zbox(\mathbf{X})\big\}.
  \end{equation}
If there exists at least one valid assignment, then this formulation is equivalent to the formulation:
  \begin{equation}
  \argmax\big\{\sum_{v\in\calV}\phi_v(x_v)~\big\vert~
  \text{$\vx$ is a valid integral assignment}\big\}.
  \end{equation}

  We may define a factor graph model for covering problems in the same manner. The only difference is in the definition of the covering factor functions, namely,
   \begin{align}\label{eqn:covering_factor}
     \psi_C(y) &\eqdf
          \begin{cases}
          0 & \text{if $\sum_{v\in\calN_G(C)} y_v \geq b_C$}\\
          \infty & \text{otherwise}
          \end{cases}
      \end{align}
  Using this factor model, we can reformulate the covering integer program \cip\ by
  \begin{equation}
  \argmin\big\{\sum_{v\in\calV}\phi_v(x_v)+\sum_{C\in\calC}\psi_C(\vx)~\big\vert~
  \vx\in \zbox(\mathbf{X})\big\}.
  \end{equation}

One could define a factor graph model for general LP's as well.
Suppose the goal is to maximize the objective function. Then, for each constraint, the range of the factor function is $\{-\infty,0\}$. If the constraint $C$ is satisfied, then the value of $\psi_C$ is $0$; otherwise it is $-\infty$.
  \section{Min-Sum Algorithms for Packing and Covering Integer
    Programs}\label{sec:min-sum}
  In this section we present the min-sum algorithm for solving packing
  and covering integer programs with zero-one constraint matrices.
  Strictly speaking, the algorithm for \pip\ is a max-sum algorithm,
  however we refer to these algorithms in general as min-sum
  algorithms. All the results in this section apply to any other
  equivalent algorithmic representation (e.g., max-product-type
  formulations). We first define the min-sum algorithms for \pip s and
  \cip s, and then state our main results.

  \subsection{The Min-Sum Algorithm}
  The min-sum algorithm for the packing integer program (\pip) is
  listed as Algorithm~\ref{alg:MP-pack}.  The input to algorithm \msp\
  consists of a factor graph model $\langle G,\Psi,\Phi,\calX\rangle$ of
  a \pip\ instance and a number of iterations $t\in\N$.  Each iteration
  consists of two parts. In the first part, each variable vertex
  performs a local computation and sends messages to all its neighboring
  constraint vertices.  In the second part, each constraint vertex
  performs a local computation and sends messages to all its neighboring
  variable vertices.
  Hence, in each iteration, two messages are sent along each edge.

  Let $\mu_{v\rightarrow C}^{(t')}(\beta)$ denote the message sent from a
  variable vertex $v\in\calV$ to an adjacent constraint vertex
  $C\in\calC$ in iteration $t'$ under the assumption that vertex $v$
  is assigned the value $\beta\in\{0,\ldots,X_v\}$.
  Similarly, let $\mu_{C\rightarrow v}^{(t')}(\beta)$ denote the message sent from
  $C\in\calC$ to $v\in\calV$ in iteration $t'$ assuming that vertex $v$
  is assigned the value $\beta\in\{0,\ldots,X_v\}$. Denote by
  $\mu_v(\beta)$ the final value computed by variable vertex $v\in\calV$
  for assignment of $\beta\in\{0,\ldots,X_v\}$.

  The initial messages (considered as the zeroth iteration) have the
  value zero and are sent along all the edges from the constraint
  vertices to the variable vertices. We refer to these initial
  messages as the \emph{zero initialization} of the min-sum algorithm.

  The algorithm proceeds with $t$ iterations. In Line~\ref{line:2a}
  the message to be sent from $v$ to $C$ is computed by adding the
  previous incoming messages to $v$ (not including the message from
  $C$) and adding to it $\phi_v(\beta)$.  In Line~\ref{line:2b} the
  message to be sent from $C$ back to $v$ is computed. The constraint
  vertex $C$ considers all the possible assignments $\vz$ to its
  neighbors in which $z_v=\beta$.  In fact, only assignments that
  satisfy the constraint of $C$ and the box constraints of the
  neighbors of $C$ are considered.  The message from $C$ to $v$ equals
  the maximum sum of the previous incoming messages (not including the
  message from $v$) among these assignments.

  Finally, in Line~\ref{line:3} each vertex $v$ decides locally on its
  outcome $\hat{x}_v$.  The maximum value $\mu_v^{\max}$ is computed,
  and the set $\delta_{v,t}$ of values that achieve the maximum is
  computed as well.  The decision $\hat{x}_v$ of vertex $v$ equals the
  minimum or maximum value in $\delta_{v,t}$ depending on the parity of
  $t$. Here, we deviate from previous descriptions of the min-sum
  algorithm that declare failure if $\delta_{v,t}$ contains more than
  one element.

    \begin{algorithm}
      \caption[\msp$(\langle G,\Psi,\Phi,\calX\rangle,t)$]{\msp$(\langle
        G,\Psi,\Phi,\calX\rangle,t)$ - A min-sum algorithm for a \pip\  $\argmax\{\vw^T\cdot \vx~\vert~\mathbf{A}\cdot \vx \leq
        \vb,~\vx\in\zbox(\calX)\}$.  Given the factor graph model $\langle
        G,\Psi,\Phi,\calX\rangle$ of the \pip\ and the number of iterations
        $t\in\N$, outputs a vector $\hat{\vx}\in \zbox (\calX)$.}
    \label{alg:MP-pack}
      \begin{enumerate}
      \item {\bf Initialize:} {\bf For each} $(v,C)\in E$ and $\beta\in\{0,\ldots,X_v\}$ {\bf do}
        \[\mu_{C\rightarrow v}^{(0)}(\beta)\gets0\]
      \item {\bf Iterate:} {\bf For} $t'=1$ to $t$ {\bf do}
        \begin{enumerate}
        \item \label{line:2a} {\bf For each} $(v,C)\in E$ and $\beta\in \{0,\ldots,X_v\}$ {\bf do} \{variable-to-constraint message\}
          \[\mu_{v\rightarrow C}^{(t')}(\beta) \gets \phi_v(\beta) + \sum_{C'\in\calN(v)\setminus\{C\}} \mu^{(t'-1)}_{C'\rightarrow v}(\beta)\]
        \item \label{line:2b} {\bf For each} $(v,C)\in E$ and $\beta\in \{0,\ldots,X_v\}$ {\bf do} \{constraint-to-variable message\}
            \[\mu_{C\rightarrow v}^{(t')}(\beta) \gets
\max
\left\{
\sum_{u\in\calN(C)\setminus\{v\}} \mu^{(t')}_{u\rightarrow C}(z_u)
\bigg\vert
\vz\in\zbox(\calX)\text{ s.t. } z_v=\beta,
\psi_C(\vz)=0
\right\}
\]

        \end{enumerate}
    \item \label{line:3} {\bf Decide:} {\bf For each} $v\in\calV$ {\bf do}
    \begin{enumerate}
      \item\label{line:3a} {\bf For each} $\beta\in\{0,\ldots,X_v\}$ {\bf do}
         \[\mu_v(\beta) \gets \sum_{C\in\calN(v)}\mu_{C\rightarrow v}^{(t)}(\beta)\]
\item $\mu^{\max}_v \triangleq \max \{\mu_v(\beta) \mid \beta\in [X_v]\}$
\item
 \label{line:3b}
 $\delta_{v,t} \eqdf \{\beta \mid \mu_v(\beta) = \mu^{\max}_v\}$.
         \item
           \begin{align*}
\hat{x}_v \gets&
          \begin{cases}
             \max\{\beta \mid \beta\in \delta_{v,t}\} & \text{if $t$ is even}\\
             \min\{\beta \mid \beta\in \delta_{v,t}\} & \text{if $t$ is odd}
          \end{cases}
           \end{align*}
      \end{enumerate}
      {\bf Return} $\hat{\vx}$
      \end{enumerate}
  \end{algorithm}

  Algorithm \msc\ listed as Algorithm~\ref{alg:MP-cover} is based on the
  following reduction of the covering LP to a packing LP as follows.
It is easy to write a direct min-sum formulation of algorithm \msc.
  \begin{claim}\label{claim:CLPtoPLP}
  Let $\vd\triangleq \mathbf{A}\cdot\calX-\vb\in\R^m$, then
  \[\argmin\big\{\vw^T\cdot \vz~\big\vert~\mathbf{A}\cdot \vz \geq \vb,~\vz\in\rbox(\calX)\big\} = \calX- \argmax\big\{\vw^T\cdot \vx~\big\vert~\mathbf{A}\cdot \vx \leq \vd,~\vx\in\rbox(\calX)\big\}.\]
  \end{claim}
  \begin{proof}
    Consider the mapping $\varphi(\vz) \triangleq \calX-\vz$.  The mapping
    $\varphi$ is a one-to-one and onto mapping from the set
    $\{\vz\in\rbox(\calX) \mid \mathbf{A}\cdot \vz \geq \vb\}$ to the set $\{\vx
    \in\rbox(\calX) \mid \mathbf{A}\cdot \vx \leq \vd\}$.  Moreover, the mapping
    $\varphi$ satisfies $\vw^T\cdot \vz = \vw^T \cdot \calX - \vw^T\cdot \varphi(\vz)$, and the claim follows.
\end{proof}

  \begin{algorithm}
    \caption[\msc$(\langle G,\Psi,\Phi,\calX\rangle,t)$]{\msc$(\langle
      G,\Psi,\Phi,\calX\rangle,t)$ - a min-sum algorithm for a \cip\
      $\argmin\{\vw^T\cdot \vz~\vert~\mathbf{A}\cdot \vz \geq \vb,~\vz\in\zbox(\calX)\}$.
      Given the factor graph model $\langle G,\Psi,\Phi,\calX\rangle$ of the
      \cip\ and the number of iterations $t\in\N$, outputs a vector
      $\hat{\vz}\in \zbox(\calX)$.}\label{alg:MP-cover}
      \begin{enumerate}
      \item  Let $\langle G,\Psi',\Phi,\calX\rangle$ denote the factor graph model for the \plp\
$$\argmax\big\{\vw^T\cdot \vx~\big\vert~\mathbf{A}\cdot \vx \leq \vd,~\vx\in\zbox(\calX)\big\},$$
where $\vd \eqdf \mathbf{A}\cdot \calX - \vb$.
    \item Let $\hat{\vx}$ denote the outcome of \msp$(\langle G,\Psi',\Phi,\calX\rangle,t)$
    \item {\bf Return} $\hat{\vz} \triangleq \calX - \hat{\vx}$.
      \end{enumerate}
    \end{algorithm}

  \subsection{Main Results}\label{subsec:divergence}
\paragraph{Notation.}
Let $\OPTLP$ denote the set of optimal solutions of the packing LP
\[
\OPTLP \eqdf \argmax\big\{\vw^T\cdot \vx~\big\vert~\mathbf{A}\cdot \vx \leq
  \vb,\vx\in\rbox(\calX)\big\}.\]
Let $\langle G,\Psi,\Phi,\calX\rangle$ denote the factor graph model of this packing LP.
Fix a variable vertex $r \in \calV$ in the factor graph $G$ of the packing LP.
Let $x^{\min}_r \triangleq \min \{ x^*_r \mid \vx^*\in \OPTLP\}$ and
let $x^{\max}_r \triangleq \max \{ x^*_r \mid \vx^*\in \OPTLP\}$.
Let $\delta^{\min}_{r,t}$ and
$\delta^{\max}_{r,t}$ denote the minimum and maximum values in
$\delta_{r,t}$, respectively.

\medskip
\noindent
The proof of the following theorem appears in Section~\ref{subsec:divergence_proof}.
\begin{theorem}[weak oscillation]\label{thm:dp}
Consider an execution of \msp$(\langle G,\Psi,\Phi,\calX\rangle,t)$.
For every variable vertex $r\in \calV$ the following holds:
    \begin{enumerate}
    \item If $t$ is even, then $\max \{\beta \mid \beta\in \delta_{r,t}\} \geq x^{max}_r$.
    \item If $t$ is odd, then $\min \{\beta \mid \beta\in\delta_{r,t}\} \leq x^{min}_r$.
    \end{enumerate}
\end{theorem}

  \begin{corollary}\label{coro:msp-divergence}
 If there exists an optimal solution $\vx\in \OPTLP$ such that $x_r$ is not an integer, then
$\hat{x}_r \geq \ceil{x_r}$ if the number of iterations $t$ is even, and
$\hat{x}_r \leq \floor{x_r}$ if $t$ is odd.
  \end{corollary}

The following corollary implies that if algorithm \msp\ outputs the same value for a vertex in two consecutive iterations, then this
value is the LP optimal value.
\begin{corollary}\label{coro:even odd}
  Let $t$ denote an even number and $s$ denote an odd number. If
  $\delta_{r,t}\cap \delta_{r,s} \neq \emptyset$, then
  $\delta_{r,t}\cap \delta_{r,s}$ contains a single element $\beta$ such that
  $x^{\min}_r= x^{\max}_r=\beta$.
\end{corollary}
\begin{proof}
  If $\beta\in \delta_{r,t}\cap \delta_{r,s}$, then by
  Theorem~\ref{thm:dp}, $x^{\max}_r \leq \delta^{\max}_{r,t} = \beta = \delta^{\min}_{r,s} \leq x^{\min}_r$.
\end{proof}

Previous works on the min-sum algorithm for optimization problems
define the case that $\delta_{r,t}$ contains more than one element as
a failure.  Under this restricted interpretation, Theorem~\ref{thm:dp}
and Corollary~\ref{coro:msp-divergence} imply a necessary condition
for the convergence of the \msp\ algorithm.  Namely, $\OPTLP$ must
contain a unique optimal solution and this optimal solution must be
integral.  Indeed, if $x^{\min}_r < x^{\max}_r$, then $\hat{x}_r$
oscillates above and below the interval $(x^{\min}_r,x^{\max}_r)$
between even and odd iterations.

Analogous results holds for covering problems. We state only the theorem that is analogous to Theorem~\ref{thm:dp}.
Redefine $\OPTLP$ so that it denotes the set of optimal solutions of the covering LP, i.e.,
\[
\OPTLP\eqdf \argmin\big\{\vw^T\cdot \vx~\big\vert~\mathbf{A}\cdot \vx \geq
  \vb,\vx\in\rbox(\calX)\big\}.\]
Let $\langle G,\Psi,\Phi,\calX\rangle$ denote the factor graph model of the covering LP.

\begin{theorem}[weak oscillation]\label{thm:dp covering}
Consider an execution of \msc$(\langle G,\Psi,\Phi,\calX\rangle,t)$.
For every variable vertex $r\in \calV$ the following holds:
    \begin{enumerate}
    \item If $t$ is even, then $\min \{\beta \mid \beta\in \delta_{r,t}\} \leq x^{min}_r$.
    \item If $t$ is odd, then $\max \{\beta \mid \beta\in\delta_{r,t}\} \geq x^{max}_r$.
    \end{enumerate}
\end{theorem}
\medskip
\noindent
See Appendix~\ref{app:convergence_proof} for a discussion of
the convergence of the min-sum algorithm.
  \section{Graph Liftings}\label{sec:graph-covers}
In this section we briefly review the definition of
graph coverings, state a combinatorial characterization
  based on~\cite{RT12}, and show how the girth can be arbitrarily
  increased.

  \subsection{Covering Maps and Liftings }

  \begin{definition}[covering\footnote{ The term covering is used
    both for optimization problems called covering problems and for
    topological mappings called covering maps.  }  map~\cite{AL02}]\label{def:covering_map}
  Let $G=(V,E)$ and $\tilde{G}=(\tilde{V},\tilde{E})$ denote finite
  graphs.  A graph homomorphism $\pi:\tilde{G}\rightarrow G$ is a
  \emph{covering map} if for every $\tilde{v} \in \tilde{V}$ the
  restriction of $\pi$ to neighbors of $\tilde{v}$ is a bijection to the
  neighbors of $\pi(\tilde{v})$.
  \end{definition}
  \noindent We refer only to finite covering maps. The pre-image
  $\pi^{-1}(v)$ of a vertex $v$ is called the \emph{fiber} of $v$. It is easy to
  see that all the fibers have the same cardinality if $G$ is connected.
  This common cardinality is called the \emph{degree} or \emph{fold
    number} of the covering map.  If $\pi: \tilde{G} \rightarrow G$ is a
  covering map, we call $G$ the \emph{base graph} and $\tilde{G}$ a
  \emph{lift} of $G$. In the case where
  the fold number of the covering map is $M$, we say that $\tilde{G}$ is
  an \emph{$M$-lift} of $G$.

    \begin{figure}
      \centering \subfigure[Base graph $G$.]{
        \includegraphics[width=0.13\textwidth]{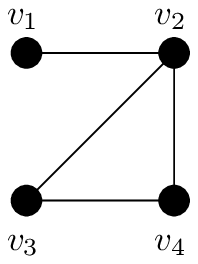}
        \label{fig:basegraph}
      }~~~~~~~~~~~~~~~ \subfigure[An $M$-lift of $G$.]{
        \includegraphics[width=0.35\textwidth]{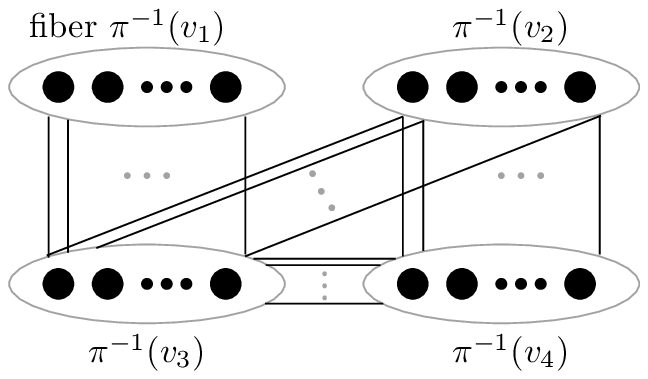}
        \label{fig:covergraph}
      }
      \caption{An $M$-lift of a base graph $G$: (1)~fiber
        $\pi^{-1}(v_i)$ consists of $M$ copies of $v_i$, and (2)~for
        each edge $(v_i,v_j)$ in $G$, the set of edges between fibers
        $\pi^{-1}(v_i)$ and $\pi^{-1}(v_j)$ is a
        matching.}\label{fig:covering}
    \end{figure}

  If $G$ is connected, then every $M$-lift of $G$ is isomorphic to an
  $M$-lift that is constructed as follows: (1)~The vertex set is simply
  $\tilde{V}\triangleq V\times[M-1]$ and the covering map is the projection defined by
  $\pi\big((v,i)\big)\triangleq v$. (2)~For every $(u,v) \in E$, the edges in
  $\tilde{E}$ between the fibers of $u$ and $v$ constitute a matching.

  The notion of $M$-lifts in graphs is extended to $M$-lifts of factor
  graph models in a natural manner.
We denote a variable vertex in
  the fiber of by $\tilde{v}$ (so $\pi(\tilde{v})$ is denoted  by $v$).
Each variable vertex $\tilde{v}$ inherits the variable
  function of $v$, namely, $\tilde{w}_{\tilde{v}} \triangleq w_v$.
  Similarly, each constraint variable $\tilde{C}$ inherits the factor
  function of $\pi(\tilde{C})$.  For brevity, we refer to the lifted
  factor graph model $\langle
  \tilde{G},\tilde{\Psi},\tilde{\Phi},\tilde{\calX}\rangle$ of
  $\langle G,\Psi,\Phi,\calX\rangle$ simply as the lift $\tilde{G}$ of
  a factor graph $G$.

  An assignment $\vx$ to the variable vertices $\calV$ of a factor
  graph $G$ is extended to an assignment $\vtx$ over the lift $\tilde{G}$ simply by defining
  $\tilde{x}_{\tilde{v}}\eqdf x_v$.  Note that this extension
  preserves the validity of assignments.

  Every assignment $\vtx$ of an $M$-lift $\tilde{G}$ induces
  an assignment of the base graph $G$ that we call the average assignment.  The \emph{average
    assignment} $\avg(\tilde{\vx})$ is defined by
  \begin{equation}
  \avg(\tilde{\vx})_v \triangleq\frac{1}{M}\cdot\sum_{\tilde{v}\in\pi^{-1}(v)}\tilde{x}_{\tilde{v}}.
  \end{equation}

Consider an
  $M$-lift $\tilde{G}$ of the factor graph $G$ and a valid integral
  assignment $\vtx$ to $\tilde{G}$.  By linearity,
  $\avg(\vtx)$ is a rational valid assignment to $G$. The following
  theorem deals with the converse situation.

  \begin{theorem}[special case of {\cite[Theorem VII.2]{RT12}}]
    \label{thm:nic} For every rational feasible solution $\vx$ of
    an LP, there exists an $M$, an $M$-lift $\tilde{G}$, and an
    integral valid assignment $\vtx$ to $\tilde{G}$ such that $\vx = \avg(\vtx)$.
  \end{theorem}
  Note that all the basic feasible solutions (extreme points) of the
  packing LP and the covering LP are rational.

  \subsection{Increasing Girth}
The following proposition deals with obtaining lifts with large girth.
\begin{proposition}\label{prop:doubleGirth}
 There exists a finite lift $\tilde{G}$ of $G$ such that $\girth(\tilde{G})\geq2\cdot\girth(G)$.
\end{proposition}
  \begin{proof}
    Given a graph $G=(V,E)$, we construct a $2^{\abs{E}}$-lift
    $\tilde{G}=(\tilde{V},\tilde{E})$ as follows.  Let $k=\abs{E}$.
    The vertices in each fiber of $\tilde{G}$ are indexed by a binary
    string of length $k$.  Index the edges in $E$ by
    $\{e_1,\ldots,e_k\}$.  For an edge $e_i=(u,v)$, the matching
    between the fiber of $u$ and the fiber of $v$ is induced simply by
    flipping the $i$'th bit in the index. Namely, $u_{\langle b_1
      \ldots b_i \ldots b_k\rangle}\mapsto v_{\langle b_1 \ldots
      \overline{b_i} \ldots b_k\rangle}$.

    Consider a cycle $\tilde{\gamma}$ in $\tilde{G}$ and its
    projection $\gamma$ in $G$. Each edge $e_i$ in $\gamma$ must
    appear an even number of times. Otherwise, the $i$'th bit is
    flipped an odd number of times in $\tilde{\gamma}$, and
    $\tilde{\gamma}$ can not be a cycle. It follows that
    $\girth(\tilde{G})\geq2\cdot\girth(G)$.
  \end{proof}

\medskip\noindent
By applying Proposition~\ref{prop:doubleGirth} repeatedly, we have the following corollary.
\begin{corollary}\label{coro:girth}
Consider a graph $G$. Then for any finite $\ell\in\N$ there exists a finite lift $\tilde{G}$ of $G$ such that $\girth(\tilde{G})\geq2^\ell$.
\end{corollary}

\section{Proof of Main Results}\label{sec:proofs}

\subsection{Min-Sum as a Dynamic Programming on Computation Trees}
Given a graph $G=(V,E)$ and a vertex $r\in V$. The path-prefix tree of height $h$ is defined as follows.
\begin{definition}[Path-Prefix Tree]\label{def:ppt}
  Let $\hat{V}$ denote the set of all paths
  with length at most $h$ without backtracks that start at vertex $r$. Let $\hat{E}
  \triangleq \big\{(p_1,p_2)\in\hat{V}\times\hat{V}\ \big\vert\ p_1\
  \mathrm{is~a~prefix~of~}p_2,~\lvert p_1\rvert+1=\lvert p_2\rvert
  \big\}$.    The
  directed graph $(\hat{V},\hat{E})$ is called the \emph{path-prefix
    tree} of $G$ rooted at vertex $r$ with height $h$, and is denoted
  by $\calT_r^{h}(G)$.
\end{definition}
\noindent
We denote the zero-length path in $\hat{V}$ by $(r)$.
The graph $\calT_r^{h}(G)$ is obviously acyclic and is an out-tree rooted at $(r)$. Path-prefix trees of $G$ that are rooted in variable vertices are often called \emph{computation trees of $G$} or \emph{unwrapped trees of $G$}.

  We use the following notation. Vertices in $\calT_r^{h}(G)$ are
  paths in $G$, and are denoted by $p$ and $q$ whereas variable
  vertices in $G$ are denoted by $u,v,r$. For a path $p\in\hat{V}$,
  let $t(p)$ denote the last vertex (i.e., target) of path $p$.

Consider a path-prefix tree $\calT_r^h(G)$ of a factor graph
$G=(\calV\cup\calC,E)$.
We denote the vertex set of $\calT_r^h(G)$ by $\hat{\calV}\cup \hat{\calC}$, where
$\hat{\calV}$ denotes paths that end in a variable vertex, and $\hat{\calC}$
denotes paths that end in a constraint vertex.
Paths in $\hat{\calV}$ are called \emph{variable paths}, and paths in
$\hat{\calC}$ are called \emph{constraint paths}.  We attach variable
functions $\hat{\phi}_p$ to variable paths $p$, and factor functions
$\hat{\psi}_q$ to constraint paths; each vertex $p$ inherits the
function of its endpoint. The box constraint for a variable path $p$ that ends at vertex $v$ is
defined by $\hat{X}_p \eqdf X_v$.

In the following lemma, the \msp\ algorithm is interpreted as a
dynamic programming algorithm over the path-prefix trees (see e.g.,
\cite[Section 2]{GSW12}).
\begin{lemma}\label{lemma:min-sumDP}
  Consider an execution of $\msp(\langle G,\Psi,\Phi,\calX\rangle,t)$.
  Consider the computation tree
  $\calT_r^{2t}(G)=(\hat{\calV}\cup\hat{\calC},\hat{E})$.  For every
  variable vertex $r\in\calV$ and $\beta\in\{0,\ldots,X_r\}$,
\begin{equation*}
\mu_r(\beta)=\max\bigg\{
\sum_{p\in\hat{\calV}}\hat{\phi}_p(\hat{z}_{p})+
\sum_{q\in\hat{\calC}}\hat{\psi}_q(\hat{z}_{\calN (q)})
\bigg\vert
~\forall p\in\hat{\calV}. \hat{z}_p\in\zbox(\hat{X}_{p}),~\hat{z}_{(r)}=\beta\bigg\}.
\end{equation*}
\end{lemma}

\begin{definition}[optimal assignment]
We say that a valid assignment $\hat{z}$ to the
variable paths in $\calT_r^{2t}(G)$ is \emph{optimal} if it maximizes the objective function
$\sum_{p\in\hat{\calV}}\hat{\phi}_p(\hat{z}_{p})+\sum_{q\in\hat{\calC}}\hat{\psi}_q(\hat{z}_{\calN_{\calT}(q)})$.
\end{definition}

Let $\OPTDP (r,t)$ denote the set of optimal valid assignments to the
variable paths in $\calT_r^{2t}(G)$.
By Line~\ref{line:3b} in algorithm \msp, the following corollary holds.
\begin{corollary}\label{cor:msp-dp}
$\delta_{r,t}= \{z_{(r)} \mid \vz\in \OPTDP(r,t)\}$.
\end{corollary}

\subsection{Weak Oscillation of \msp\ - Proof of Theorem~\ref{thm:dp}}\label{subsec:divergence_proof}

\begin{proof}[Proof of Theorem~\ref{thm:dp}]
We prove Part (1) of the Theorem.
  Fix an assignment $\vz\in \OPTDP(r,t)$ such that
  $z_{(r)}=\delta^{\max}_{r,t}$.  Fix an optimal solution $\vx\in \OPTLP$
  such that $x_r=x^{\max}_r$.  Assume, for the sake of contradiction,
  that $t$ is even and that $z_{(r)} < x^{\max}_r$.

  Without loss of generality, $\vx$ is a basic feasible solution.  By Theorem~\ref{thm:nic} and
  Corollary~\ref{coro:girth}, for some $M$, there exists an $M$-lift
  $\langle \tilde{G},\tilde{\Psi},\tilde{\Phi},\tilde{\calX}\rangle$
  of $\langle G,\Psi,\Phi,\calX\rangle$ such that: (i)~there exists an
  integral valid assignment $\vtx$ for $\tilde{G}$ such that
  $\vx=\avg(\vtx)$, and (ii)~$\girth(\tilde{G})> 4t$.

  The value of $x_r$ equals the average of $\tilde{\vx}$ over the fiber of $r$.
  Let $\tilde{r}$ denote a vertex in the fiber of $r$ such
  that $\tilde{x}_{\tilde{r}}\geq x^{\max}_r$.

  Let $B_{\tilde{G}}(\tilde{r},2t)$ denote the ball of radius $2t$
  centered at $\tilde{r}$. Denote by $\tilde{G}_{B(\tilde{r},2t)}$ the
  subgraph of $\tilde{G}$ induced by $B_{\tilde{G}}(\tilde{r},2t)$.
  Because $\girth(\tilde{G})>4t$, $\tilde{G}_{B(\tilde{r},2t)}$ is a
  tree. It follows that $\tilde{G}_{B(\tilde{r},2t)}$ is isomorphic to
  the computation tree $\calT_r^{2t}(G)$.  Because $\vz$ is an optimal
  valid assignment to $\calT_r^{2t}(G)$, we can regard $\vz$ also as
  an optimal valid assignment to the variable vertices in
  $\tilde{G}_{B(\tilde{r},2t)}$.

  Because $z_{\tilde{r}} < x^{\max}_r$, the restriction of $\vtx$ to the
  variable vertices in $\tilde{G}_{B(\tilde{r},2t)}$ does not equal
  $\vz$. Our goal is to obtain a contradiction by showing that either
  $\vz$ is not an optimal assignment or there exists an optimal
  solution $\mathbf{y}\in \OPTLP$ such that $y_r > x^{\max}_r$. We
  show this by constructing ``hybrid'' integral solutions.

  Let $\calE \triangleq \{\tilde{u} \in B(\tilde{v},2t) \mid
  \frac{1}{2}\cdot d(\tilde{u}, \tilde{r}) \text{ is even}\}$.
  Similarly, let $\calO \triangleq \{\tilde{u} \in B(\tilde{r},2t)
  \mid \frac{1}{2} \cdot d(\tilde{u}, \tilde{r}) \text{ is odd}\}$.  Note
  that both $\calE$ and $\calO$ contain only variable vertices.  We
  refer to $\calE$ as the \emph{even layers} and to $\calO$ as the
  \emph{odd layers}.

Let $\calF$ denote the subgraph of
  $\tilde{G}_{B(\tilde{r},2t)}$ that is induced by:
\begin{inparaenum}[(i)]
\item the vertices $\tilde{u}\in \calE$ such that $z_{\tilde{u}} < \tilde{x}_{\tilde{u}}$,
\item the vertices $\tilde{u} \in \calO$ such that $z_{\tilde{u}} > \tilde{x}_{\tilde{u}}$, and
\item constraint vertices in $\tilde{G}_{B(\tilde{r},2t)}$.
\end{inparaenum}
By definition, $\tilde{r}$ is a vertex in $\calE$, and
by our assumption $z_{\tilde{r}}<\tilde{x}_r$. Hence, $\tilde{r}\in \calF$. Let $\calT$
denote the connected component of $\calF$ that contains $\tilde{r}$.
We root $\calT$ at $\tilde{r}$, and refer to $\calT$ as an
\emph{alternating tree}.

A subtree $\calT_S$ of $\calT$ is a \emph{skinny tree} if each
constraint vertex chooses only one child and each variable vertex
chooses all its children. Formally, a subtree $\calT_S$ of $\calT$ is
a skinny tree if it is a maximal tree with respect to inclusion among all
trees that satisfy (i)~$\tilde{r}\in \calT_S$,
(ii)~$\deg_{\calT_S}(\tilde{C})=2$ for every constraint vertex
$\tilde{C}$ in $\calT_S$, and
(iii)~$\deg_{\calT_S}(\tilde{u})=\deg_{\calT}(\tilde{u})$ for every
variable vertex $\tilde{u}$ in $\calT_S$.  We fix a skinny subtree
$\calT_S$ of $\calT$ to obtain a contradiction.

For a subset of variable vertices $\tilde{A}\subseteq \tilde\calV$,
let $\tilde{w}(\tilde{A})\triangleq \sum_{\tilde{u}\in \tilde{A}}
w_u$ (recall that the weight $w_u$ of a
variable vertex $u$ in  $G$ is given to each vertex $\tilde{u}$ in the fiber of $u$) .  We claim that
\begin{align}
  \label{eq:EgeqO}
  \tilde{w}\big(\TScapE\big) &\geq\tilde{w}\big(\TScapO\big),
\end{align}

To prove Equation~(\ref{eq:EgeqO}), define an integral assignment $\vy$ to variable vertices in $\tilde{\calV}$ by
\begin{equation*}
\tilde{y}_{\tilde{u}} \triangleq \begin{cases}
\tilde{x}_{\tilde{u}} -1 &\mathrm{if}~\tilde{u}\in\TScapE\\
\tilde{x}_{\tilde{u}} +1 &\mathrm{if}~\tilde{u}\in\TScapO\\
\tilde{x}_{\tilde{u}}    &\mathrm{otherwise.}\\
\end{cases}
\end{equation*}
Observe that $\vy$ is a valid assignment for $\tilde{G}$.  Indeed, all
the box constraints are satisfied because we increment a value
compared to $\tilde{x}_u$ only if
$\tilde{x}_{\tilde{u}}<z_{\tilde{u}}$.  Similarly, we decrement a
value compared to $\tilde{x}_u$ only if
$\tilde{x}_{\tilde{u}}>z_{\tilde{u}}$.  In addition, we need to show
that every constraint is satisfied by $\vy$. Note that a constraint
$\tilde{C}$ may have at most two neighbors that are variable vertices
in the skinny tree; the rest of the neighbors retain the value
assigned by $\vtx$.  If a constraint $\tilde{C}$ is not a neighbor of
a variable vertex in $\calT_S$, then it is satisfied because $\vtx$
satisfies it.  If a constraint $\tilde{C}$ has two neighbors that are
variable vertices in $\calT_S$, then one is incremented and the other
is decremented.  Overall, the constraint remains satisfied.  Finally,
suppose $\tilde{C}$ has only one neighbor in $\calT_S$. Denote this
neighbor by $\tilde{v}$. Then $\tilde{v}$ is a parent of $\tilde{C}$.
If $\tilde{y}_{\tilde{v}}=\tilde{x}_{\tilde{u}}-1$, then clearly
$\tilde{C}$ is satisfied by $\vy$.  If the value assigned to
$\tilde{v}$ is incremented, then $\tilde{v}\in\calO$ (i.e., an odd
layer).  This implies that the children of $\tilde{C}$ are in an even
layer, the distance of which to the root is at most $2t$. Hence, the children of $\tilde{C}$ belong to the ball
$B(\tilde{r},2t)$. Moreover, these children do not belong to the alternating tree (otherwise, one of its
  children would belong to the skinny tree).  Thus, for each child
  $\tilde{u}$ of $\tilde{C}$ we have $z_{\tilde{u}} \geq
  \tilde{x}_{\tilde{u}}=\tilde{y}_{\tilde{u}}$.  In addition,
  $z_{\tilde{v}} \geq \tilde{y}_{\tilde{v}}$.  Hence, $\vz \geq \vy$
  when restricted to the neighbors of $\tilde{C}$. Because $\vz$
  satisfies $\tilde{C}$, so does $\vy$, as required.
Because $\vy$ is a valid assignment, $\avg(\vy)$ is a feasible
solution of the packing LP.  The optimality of $\vx$ implies that
$\vw^T\cdot \vx \geq \vw^T \cdot
\avg(\vy)$.
By the definition of $\vy$, we have
\[
\vw^T\cdot \big(\vx-\avg(\vy)\big)=\frac{1}{M} \cdot \big(\tilde{w}(\TScapE) -
\tilde{w}( \TScapO)\big),
\] and Equation~(\ref{eq:EgeqO}) follows.


\noindent We now define an assignment $\theta$ to variable vertices in $\tilde{G}_{B(\tilde{r},2t)}$ by
\begin{equation*}
\theta_{\tilde{u}} \triangleq \begin{cases}
z_{\tilde{u}} +1 &\mathrm{if}~\tilde{u}\in\TScapE\\
z_{\tilde{u}} -1 &\mathrm{if}~\tilde{u}\in\TScapO\\
z_{\tilde{u}}    &\mathrm{otherwise}\\
\end{cases}
\end{equation*}
We claim that $\theta$ is a valid integral assignment for
$\tilde{G}_{B(\tilde{r},2t)}$.  The proof is analogous to the proof
that $\vy$ is a valid assignment.
By Equation~(\ref{eq:EgeqO}),
the value of $\theta$ is not less than the value of $\vz$ since
\begin{equation}
\sum_{\tilde{u}\in \tilde\calV\cap B_{\tilde{G}}(\tilde{r},2t)} \big(\phi_{\tilde{u}} (\theta_{\tilde{u}}) -
\phi_{\tilde{u}} (z_{\tilde{u}})\big)
=\tilde{w}\big(\TScapE\big)-\tilde{w}\big(\TScapO\big) \geq 0.
\end{equation}
Therefore, $\theta\in \OPTDP (r,t)$.
However, $\theta_r > z_r =
\delta^{\max}_{r,t}$, a contradiction.  It follows that
$\delta^{\max}_{r,t} \geq x^{\max}_{r,t}$ if $t$ is even.

The proof of Part (2) that $\delta^{\min}_{r,t} \leq x^{\min}_{r,t}$ for an odd $t$  is analogous 
to the proof that $\delta^{\max}_{r,t} \geq x^{\max}_{r,t}$ if $t$ is even.
It requires the following
modifications.
\begin{inparaenum}[(1)]
\item Fix $\vz\in \OPTDP(r,t)$ such that $z_{(r)}=\delta^{\min}_r$ and $\vx\in\OPTLP$ such that $x_r=x^{\min}_r$.
\item Assume towards a contradiction that $t$ is odd and $z_{(r)}>x_r$.
\item Pick $\tilde{r}$ so that $\tilde{x}_{\tilde{r}} \leq x_r$.
\item The forest $\calF$ is induced by the following set of vertices:
\begin{inparaenum}[(i)]
\item the vertices $\tilde{u}\in \calE$ such that $z_{\tilde{u}} > \tilde{x}_{\tilde{u}}$,
\item the vertices $\tilde{u} \in \calO$ such that $z_{\tilde{u}} < \tilde{x}_{\tilde{u}}$, and
\item constraint vertices in $\tilde{G}_{B(\tilde{r},2t)}$.
\end{inparaenum}
\item Prove that the weight of even layers in the skinny tree is not greater than the weight of the odd layers.
\item The assignment $\vy$ is defined so that it increments even layers and decrements odd layers.
\item The assignment $\theta$ is defined so that it decrements even layers and increments odd layers.
\end{inparaenum}
\end{proof}

\section*{Acknowledgments}
The authors would like to thank Nicholas Ruozzi and Kamiel Cornelissen for helpful comments.

\bibliographystyle{alpha}

\appendix

\section{On Convergence of the Min-Sum Algorithm for Nonbinary Packing and Covering Problems}\label{app:convergence_proof}

In Section~\ref{subsec:divergence} we showed that if the \msp\
algorithm outputs the same result in two consecutive iterations, then this result equals the optimal solution of the LP relaxation (see Corollary~\ref{coro:even odd}).
On the other hand,
even if the LP relaxation has a unique optimal solution and that solution is integral,
then the \msp\ algorithm may not converge (see Sanghavi \emph{et al.}~\cite{SSW09}
for an example with respect to the maximum weight independent set problem).

Convergence of the min-sum algorithm was proved for the maximum weight
$b$-matching and the minimum $r$-edge covering problems by Sanghavi
\emph{et al.} \cite{SMW11} and Bayati \emph{et al.} \cite{BBCZ11}.
They considered the zero-one integer program for maximum weight
matching with the constraints $\sum_{e\ni v} x_e \leq 1$, and
proved that after a pseudo-polynomial number of iterations, the
min-sum algorithm converges to the optimal solution of the LP
relaxation provided that it is unique and integral.  The parameter
that is used to bound the number of iterations is defined as follows.
\begin{definition}[\cite{SMW11}]\label{def:c}
    Given a polyhedron $\calP\subseteq \R^n$ and a cost vector $\vw\in
    \R^n$. Define $c(\calP,\vw)$ by
    \begin{align*}
      c(\calP,\vw) &\eqdf \min_{x\in \calP\setminus\{x^*\}}\frac{\vw^T\cdot(x^*-x)}{\| x^*-x\|_1},
    \end{align*}
  where $x^*=\argmax\{\vw^T \cdot x \mid x\in \calP\}$.
  \end{definition}
By definition, $C(\calP,\vw)\geq 0$, and $c(\calP,\vw)>0$ if and only the LP has a unique optimal solution.
On the other hand, $c(\calP,\vw)\leq w_{\max}$, where $w_{\max}\triangleq\max\{w_i\mid 1\leq i\leq n\}$.

We generalize the convergence result of Sanghavi \emph{et al.}~\cite{SMW11} to nonbinary packing and covering problems.
For the sake of brevity we state the result for packing problems; the analogous result for covering problems holds as well.
\begin{theorem}\label{thm:msp-convergence}
Let $\calP$ denote the polytope $\{\vx\in\rbox(\calX)\mid \mathbf{A}\cdot \vx\leq \vb\}$.
Assume that every column of $\mathbf{A}$ contains at most two $1$s.
Assume that the packing LP $\argmax\{\vw^T\cdot \vx~\vert~\vx\in\calP\}$ has a unique optimal solution $\vx^*$ such that $x^*_i\in\{0,X_i\}$ for every $1\leqslant i\leqslant n$.
Let $\langle G,\Psi,\Phi,\calX\rangle$ denote the factor graph model of the packing LP.
If $t> \frac{w_{\max}}{c(\calP,w)}+\frac{1}{2}$, then the output $\hat{\vx}$ of
Algorithm \msp$(\langle G,\Psi,\Phi\rangle,t)$ satisfies $\hat{\vx}=\vx^*$.
\end{theorem}
\medskip
\noindent
We first explain why the two restricting assumptions in Theorem~\ref{thm:msp-convergence} are useful.
\begin{observation}\label{obs:path}.
If each column of $\mathbf{A}$ contains at most two $1$s, then the degrees of the variable vertices in the factor graph are at most $2$.
Hence, the alternating skinny tree (as in the proof of Theorem~\ref{thm:dp}) reduces to a path.
\end{observation}
For every $M$-lift $\tilde{G}$ of a factor graph $G$, we define
the polytope $\tilde{\calP} = \{\tilde{\vx}\in\rbox(\tilde{\calX})\mid
\tilde{\mathbf{A}}\cdot \tilde{\vx}\leq \tilde{\vb}\}$, where
$\tilde{\mathbf{A}}$ and $\tilde{\vb}$ are the constraint matrix and
vector of the lifted factor graph.
\begin{observation}\label{obs:x}
If the unique solution $\vx^*$ of the packing LP satisfies $x^*_i\in\{0,X_i\}$ for every $1\leqslant i\leqslant n$,
then $C(\tilde\calP,\tilde\vw)= C(\calP,\vw)$ for every $M$-lift $\tilde{G}$ of the factor graph $G$.
\end{observation}
\begin{proofof}{Theorem~\ref{thm:msp-convergence}}
  We focus on the case that $t$ is even; the proof for odd values of
  $t$ is analogous.  The proof uses many of the notions used in the
  proof of Theorem~\ref{thm:dp} with modifications based on
  Observations~\ref{obs:path} and~\ref{obs:x}, hence we reuse the
  same notations.  Note that if $t$ is even and $x^*_r=X_r$, then
  Theorem~\ref{thm:dp} implies that $\delta^{\max}_{r,t}=X_r$, and
  hence $\hat{x}_r=x^*_r$, as required.  Thus, we are left with the
  case that $x^*_r=0$ and wish to prove that $\delta^{\max}_{r,t}=0$
  if $t$ is even.

  Assume towards a contradiction that $\delta^{\max}_{r,t} >0$, and
  let $\vz\in\OPTDP(r,t)$ denote an assignment such that $z_{(r)}=\delta^{\max}_{r,t} >0$.  As
  in the proof of Theorem~\ref{thm:dp}, we define an alternating tree
  in $\tilde{G}_{B(\tilde{r},2t)}$. However, the variable vertices in
  the even layers of the alternating tree satisfy $z_{\tilde{v}} >
  \tilde{x}^*_{\tilde{v}}$.  Variables vertices in the odd layers
  of the alternating tree satisfy $z_{\tilde{v}} <
  \tilde{x}^*_{\tilde{v}}$.  By Obs.~\ref{obs:path}, the skinny tree
  is simply a path that contains the root $\tilde{r}$.

Let $L$ denote the set of leaves of the skinny tree whose distance from $\tilde{r}$ is $2t$.
Because the skinny tree is a path, it follows that $|L|\leq 2$.

Define an integral assignment $\vy$ to variable vertices in $\tilde{\calV}$ by
\begin{equation*}
\tilde{y}_{\tilde{u}} \triangleq \begin{cases}
\tilde{x}^*_{\tilde{u}} +1 &\mathrm{if}~\tilde{u}\in(\TScapE)\setminus L\\
\tilde{x}^*_{\tilde{u}} -1 &\mathrm{if}~\tilde{u}\in\TScapO\\
\tilde{x}^*_{\tilde{u}}    &\mathrm{otherwise.}\\
\end{cases}
\end{equation*}
As in the proof of Theorem~\ref{thm:dp}, the assignment $\vy$ is a
valid assignment for $\tilde{G}$.

Assume that $L=\emptyset$.
Since $\tilde{x}^*$ is a unique optimal assignment, it follows that
\begin{align}
  \label{eq:OgeqEL}
\tilde{w}\big(\TScapO\big) &> \tilde{w}\big((\TScapE)\setminus L\big).
\end{align}
Define
\begin{equation}\label{eq:theta2}
\theta_{\tilde{u}} \triangleq \begin{cases}
z_{\tilde{u}} - 1 &\mathrm{if}~\tilde{u}\in\TScapE\\
z_{\tilde{u}} +1 &\mathrm{if}~\tilde{u}\in\TScapO\\
z_{\tilde{u}}    &\mathrm{otherwise.}\\
\end{cases}
\end{equation}
As in the proof of Theorem~\ref{thm:dp}, the assignment $\theta$ is a valid integral assignment for
$\tilde{G}_{B(\tilde{r},2t)}$. But, Equation~(\ref{eq:OgeqEL}) implies that the assignment $\theta$ has a higher value than $\vz$. This
contradicts the optimality of $\vz$.

Assume that $|L|=2$ (the proof for $|L|=1$ is similar).
By the definition of $C(\tilde{\calP},\tilde{\vw})$ and by Obs.~\ref{obs:x},
\begin{align}
  \label{eq:cpw}
  \frac{\tilde{\vw}^T\cdot(\tilde{\vx}^*-\vy)}{\| \tilde{\vx}^* -  \vy \|_1} \geq C(\calP,\vw).
\end{align}
Note that
\begin{align*}
  \tilde{\vw}^T\cdot(\tilde{\vx}^*-\vy) &=  \tilde{w}\big(\TScapO\big)-\tilde{w}\big((\TScapE)\setminus L\big),~\mathrm{and}\\
\| \tilde{\vx}^* -  \vy \|_1 &= |\calT_S| - |L| =2t- 1.
\end{align*}
It follows that
\begin{align}
  \tilde{w}\big(\TScapO\big)-\tilde{w}\big((\TScapE)\setminus L\big) &\geq (2t-1)\cdot
  C(\calP,\vw).
\label{eq:OgeqEL2}
\end{align}
Consider the assignment $\theta$ defined in Equation~(\ref{eq:theta2}).
By  Equation~(\ref{eq:OgeqEL2}) we have
\begin{align*}
  \sum_{\tilde{u}\in \tilde\calV\cap B_{\tilde{G}}(\tilde{r},2t)}
  (\phi_{\tilde{u}} (\theta_{\tilde{u}}) - \phi_{\tilde{u}}
  (z_{\tilde{u}})) &=\tilde{w}\big(\TScapO\big)-\tilde{w}\big(\TScapE\big)\\
  &=\tilde{w}\big(\TScapO\big)-\tilde{w}\big((\TScapE)\setminus L\big) - \tilde{w}\big(L\big)\\
  &\geq (2t-1)\cdot C(\calP,\vw) - 2\cdot w_{\max}.
\end{align*}
Because $t> \frac{w_{\max}}{c(\calP,w)} +\frac{1}{2}$, it follows that
the assignment $\theta$ has a higher value than $\vz$. This
contradicts the optimality of $\vz$. It follows that $\delta^{\max}_{r,t}=0$ if $t$ is even.
\end{proofof}
\end{document}